%% file: main.tex
\documentclass[letterpaper,10pt, conference]{ieeeconf}
\IEEEoverridecommandlockouts
\usepackage{cite}
\usepackage{amsfonts,amssymb,amsmath}
\usepackage{graphics}
\usepackage{subfigure}
\usepackage{tikz}
\usepackage{enumerate}

\newtheorem{theorem}{Theorem}
\newtheorem{definition}{Definition}
\newtheorem{lemma}{Lemma}

\newtheorem{remark}{Remark}

\newlength\figureheight
\newlength\figurewidth

\DeclareMathOperator*{\argmin}{arg\; min}     

\newcommand{\I}{{{\rm I}}}

\newcommand{\card}[1]{\left|#1\right|}
\newcommand{\ie}{{\it i.e.}}
\newcommand{\norm}[1]{\left\lVert#1\right\rVert}
\newcommand{\T}{^{\top}}

\title{Convex Optimization Based State Estimation against Sparse Integrity Attacks}

\author{Duo Han$^*$, Yilin Mo$^*$ and Lihua Xie$^*$
  \thanks{$*$: School of Electrical and Electronic Engineering, Nanyang Technological University, Singapore. Email: {dhanaa,ylmo,elhxie}@ntu.edu.sg}
}

\begin{document} \maketitle

\begin{abstract}
We consider the problem of robust estimation in the presence of integrity attacks. There are $m$ sensors monitoring the state and $p$ of them are under attack. The malicious measurements collected by the compromised sensors can be manipulated arbitrarily by the attacker. The classical estimators such as the least squares estimator may not provide a reliable estimate under the so-called $(p,m)$-sparse attack. In this work, we are not restricting our efforts in studying whether any specific estimator is resilient to the attack or not, but instead we aim to present some generic sufficient and necessary conditions for robustness by considering a general class of convex optimization based estimators. The sufficient and necessary conditions are shown to be tight, with a trivial gap.
\end{abstract}

\input{intro.tex}

\section{Problem Setup}\label{section:problemsetup}
\subsection{System Model}
Assume that $m$ sensors are measuring the state $x$ and the measurement equation for the $i$th sensor is given by
\begin{align}\label{eq:goodsystem}
  z_i = H_i x + w_i,
\end{align}
where $x\in\mathbb R^n$ is the state of interest, $z_i\in \mathbb R^{m_i}$ is the ``true'' measurement collected by the $i$th sensor, and $w_i \in \mathbb R^{m_i}$ is the measurement noise for the $i$th sensor. The measurement matrix $H\triangleq [H_1\T,H_2\T,\ldots,H_i\T]\T\in\mathbb R^{(\sum_i m_i)\times n}$ is assumed to be observable, \ie, $H$ is full column rank. In the presence of attacks, the measurement equation can be written as
\begin{align}\label{eq:badsystem}
  y_i = z_i + a_i = H_i x + w_i + a_i,
\end{align}
where $y_i\in \mathbb R^{m_i}$ is the ``manipulated'' measurement and $a_i\in \mathbb R^{m_i}$ is the attack vector. In other words, the attacker can change the measurement of the $i$th sensor by $a_i$. Denote
\begin{align}
 z&\triangleq[z_1\T,z_2\T,\ldots,z_m\T]\T, &
y&\triangleq[y_1\T,y_2\T,\ldots,y_m\T]\T, \\
w&\triangleq[w_1\T,w_2\T,\ldots,w_m\T]\T,&
 a&\triangleq[a_1\T,a_2\T,\ldots,a_m\T]\T.\nonumber
\end{align}

Denote the index set of all sensors as $\mathcal S\triangleq\{1,2,\ldots,m\}$. For any index set $\mathcal I\subseteq \mathcal S$, define the complement set to be $\mathcal I^c\triangleq \mathcal S\backslash\mathcal I$. In our attack model, we assume that the attacker can only compromise at most $p$ sensors but can arbitrarily choose $a_i$. Formally, a $(p,m)$-sparse attack can be defined as
\begin{definition}[$(p,m)$-sparse attack]
A vector $a$ is called a $(p,m)$-sparse attack if there exists an index set $\mathcal I\subset \mathcal S$, such that:
\begin{enumerate}[(i)]
  \item $\norm{a_i} = 0,~\forall i\in \mathcal I^c ;$
  \item $\card{\mathcal I} \leq p.$
\end{enumerate}
\end{definition}
Define the collection of a possible index set of malicious sensors as
\[
\mathbb C\triangleq\{\mathcal I:\mathcal I\subset\mathcal S,\card{\mathcal I} = p\}.
\]
 The set of all possible $(p,m)$-sparse attacks is denoted as
\[
\mathcal A \triangleq \bigcup_{\mathcal I\in\mathbb C}\{a: \norm{a_i} = 0, i\in \mathcal I^c\}.
\]

The main task of this work is to investigate the generic sufficient and necessary conditions for an estimator to be robust to $(p,m)$-sparse attacks. To this end, we first formally define the robustness of an estimator.

\begin{definition}[Robustness]
An estimator $g:\mathbb R^{\sum_i m_i }\mapsto \mathbb R^n$ which maps the measurements $y$ to a state estimate $\hat x$ is said to be robust to the $(p,m)$-sparse attack if it satisfies the following condition:
\begin{align}
\norm{g(z)-g(z+a)}\leq \mu(z),~\forall a\in\mathcal A,\label{eq:defrobust}
\end{align}
where $\mu:\mathbb R^{\sum_i m_i}\mapsto \mathbb R$ is a real-valued mapping on $z$.
\end{definition}
The robustness implies that the disturbance on the state estimate caused by an arbitrary attack is bounded. A trivial robust estimator is $g(y)=0$ which provides very poor estimate. Therefore, another desirable property for an estimator is translation invariance, which is defined as follows:
\begin{definition}[Translation invariance]
An estimator $g$ is translation invariant if $g(z+Hu)=u+g(z),~\forall u\in\mathbb R^n$.
\end{definition}
\begin{remark}
  Notice that if an estimator is robust and translation invariant, then
  \begin{align*}
    \| g(z) - g(z+a) \|&= \|x + g(w) - x + g(w+a) \| \\
                       &= \|g(w)-g(w+a)\|\leq \mu(w).
  \end{align*}
  Therefore, the maximum bias that can be injected by an adversary is only a function of the noise $w$.
\end{remark}
In the next subsection, we propose a general convex optimization based estimator which is translation invariant.

\subsection{A General Estimator}
A large variety of estimators are developed by the research community to solve the state estimation problem. In order to achieve greater generality, we first propose a general convex optimization based estimator. We then show that many estimators can be rewritten in this general framework.

The estimator that we study in this paper is assumed to have the following form:
\begin{align}
  \hat x = g(y) \triangleq \argmin_{\hat x} \sum_{i\in\mathcal S} f_i(y_i-H_i\hat x),
  \label{eq:general}
\end{align}
where the following properties of function $f_i:\mathbb R^{m_i}\mapsto \mathbb R$ are assumed:
\begin{enumerate}[(i)]
\item $f_i$ is convex.
\item $f_i$ is symmetric, \ie, $f_i(u) = f_i(-u)$.
\item $f_i$ is non-negative and $f_i(0) = 0$.
\end{enumerate}

\begin{remark}\label{remark:1}
  It is easy to check that the estimator $g$ is translation invariant. One can view $y_i-H_i\hat x$ as the residue for the $i$th sensor and $f_i$ as a cost function. The convex constraints on $f_i$ ensures that the minimization problem can be solved in an efficient (possibly also distributed) way. The symmetric assumption on $f_i$ is typically true for many practically used estimator and can actually be relaxed. The last assumption implies that the cost achieves minimum value when the residue is $0$.
\end{remark}
We now investigate several commonly used estimator and show that they can be written as \eqref{eq:general}.
\begin{enumerate}[(a)]
\item Least Square Estimator:
  \begin{align}
    \hat x &=  \argmin_{\hat x} \norm{y-H\hat x}_2^2=  \argmin_{\hat x} \sum_{i\in\mathcal S}\norm{y_i-H_i\hat x}_2^2\nonumber\\
           & = (H\T H)^{-1}H\T y. \label{eq:lsestimator}
  \end{align}

\item Another example is an estimator which minimizes the sum of the $l_1$ norm of the residue, \ie,
  \begin{align}
    \hat x = \argmin_{\hat x} \sum_{i\in\mathcal S}\norm{y_i-H_i\hat x}_1.  \label{eq:medianestimator}
  \end{align}
  In the case that $m_i=n$ and $H_i=\I_n,~\forall i$, the estimate is a vector in which the $i$th entry is the median over the $i$th entries of all measurements $y_i$'s.

\item The following is designed to minimize the sum of the $l_2$ norm of the residue:
  \begin{align}
    \hat x = \argmin_{\hat x} \sum_{i\in\mathcal S}\norm{y_i-H_i\hat x}_2.  \label{eq:geometricmedianestimator}
  \end{align}
  The optimal estimate in the case that $m_i=n$ and $H_i=\I_n,~\forall i$ is the geometric median of all $y_i$'s, which is called an $L_1$ estimator in \cite{Lopuhaa1991}. In other words, $\hat x$ is the point in $\mathbb R^n$ that minimizes the sum of Euclidean distances from $y_i$ to that point.

\item Pajic et al.~\cite{Pajic2014} proposed the following robust estimator in the presence of integrity attack:
  \begin{align*}
    & \mathop{\textit{minimize}}\limits_{\hat x,a,w}&
    & \|w\|^2\\
    &\text{subject to}&
    &y = H\hat x + w + a,\,\|a\|_0\leq q.
  \end{align*}
  However, the minimization problem involves zero-norm, and thus is difficult to solve in general. A commonly adopted approach is to use $L_1$ relaxation to approximate zero-norm, which leads to the following minimization problem:
  \begin{align}
    & \mathop{\textit{minimize}}\limits_{\hat x,a,w}&
    & \|w\|^2+\lambda \|a\|_1\label{eq:optlasso}\\
    &\text{subject to}&
    &y = H\hat x + w + a.\nonumber
  \end{align}
  If we define the following function:
  \begin{align}
    d(u)~\triangleq ~&\mathop{\textit{minimize}}\limits_{a_i}&
    &  \norm{u-a_i}_2^2 + \lambda  \norm{a_i}_1 \label{eq:lasso2}
  \end{align}
  Then one can easily prove that the optimization problem \eqref{eq:optlasso} can be rewritten as
  \begin{align}
    \hat x = \argmin_{\hat x} \sum_{i\in\mathcal S} d(y_i - H_i \hat x).\label{eq:lasso}
  \end{align}
\end{enumerate}

In the next section, we shall present sufficient and necessary conditions for the robustness of the general estimator \eqref{eq:general}. Since \eqref{eq:medianestimator},~\eqref{eq:geometricmedianestimator} and \eqref{eq:lasso} are all special cases of \eqref{eq:general}, we can easily analyze their individual robustness.

\section{Robust Analysis for a General Estimator}\label{section:main}
This section is devoted to the derivation of necessary and sufficient conditions for the robustness of the general estimator.
Denote the compact set $\mathcal U\triangleq\{u\in \mathbb R^n:\norm{u}= 1\}$. Before proceeding to the main results, we need the following lemma.
\begin{lemma}\label{lemma:convex}
  Let $q:\mathbb R\rightarrow \mathbb R$ be a convex function and $q(0) = 0$, then $q(t)/t$ is monotonically non-decreasing on $t\in \mathbb R^+$. Moreover,
  \begin{align}
    q(t+1)-q(t)\geq q(t)/t.
    \label{eq:marginalincrease}
  \end{align}
\end{lemma}
\begin{proof}
  For any $0 < \alpha < 1$, we have
  \begin{align*}
    q(\alpha t) \leq \alpha q(t) + q(0) = \alpha q(t).
  \end{align*}
  Divide both side by $\alpha t$, we can prove that $q(t)/t$ is monotonically non-decreasing. Therefore, $q(t+1)/(t+1)\geq q(t)/t$, which implies \eqref{eq:marginalincrease}.
\end{proof}
As a consequence of Lemma~\ref{lemma:convex}, we know that $f_i(tH_iu)/t$ is monotonically non-decreasing. As a result, there are only two possibilities:
\begin{enumerate}[(i)]
\item $f_i(tH_iu)/t$ is bounded for all $i$ and for all $u$, which implies that the limit $\lim_{t\rightarrow\infty}f_i(tH_iu)/t$ exists.
\item $f_i(tH_iu)/t$ is unbounded for some $i$ and $u$.
\end{enumerate}
The next lemma provides several important properties for the case where $\lim_{t\rightarrow\infty}f_i(tH_iu)/t$ exists, whose proof is reported in the appendix:
\begin{lemma}\label{lemma:1}
  If the following limit is well defined, \ie, finite, for all $u\in \mathbb R^{n}$:
  \begin{align}
    \lim_{t\rightarrow\infty} \frac{f_i(tH_iu)}{t} = C_i(u),
    \label{eq:sublinear}
  \end{align}
  then the following statements are true:
  \begin{enumerate}[(i)]
  \item $C_i(\alpha u)=\card{\alpha} C_i(u)$ and $C_i(u_1+u_2)\leq C_i(u_1)+C_i(u_2)$.
  \item Define the function $h_i(u,v,t):\mathbb R^n\times\mathbb R^{m_i}\times \mathbb R\mapsto\mathbb R$,
    \begin{align}
      h_i(u,v,t) \triangleq \frac{1}{t}\left[f_i(v+tH_iu)-f_i(v)\right].\label{eq:defhfunction}
    \end{align}
    Then the following pointwise limit holds:
    \begin{align}
      \lim_{t\rightarrow\infty}h_i(u,v,t)  = C_i(u).
      \label{eq:deltalimit}
    \end{align}
    Moreover, the convergence is uniform on any compact set of $(u,v)$.
  \item For any $v$ and $u$, we have that
    \begin{align}
      f_i(v+H_iu) - f_i(v) \leq C_i(u).
      \label{eq:deltaw}
    \end{align}
  \end{enumerate}
\end{lemma}
\begin{remark}\label{remark:force}
Intuitively speaking, one can interpret $f_i$ as a potential field and the derivative of $f_i$ as the force generated by sensor $i$ (if it is differentiable). By \eqref{eq:deltaw}, we know that the force from the potential field $f_i$ along the $u$ direction cannot exceed $C_i(u)$ (or $C_i(u)/\|u\|$ to normalize). On the other hand, Equation~\eqref{eq:deltalimit} implies that this bound is achievable.
\end{remark}

We now give the sufficient condition for the robustness of the estimator.
\begin{theorem}[Sufficient condition]
  If the following conditions hold:
  \begin{enumerate}
  \item $C_i(u)$ is well defined for all $u\in \mathbb R^{n}$ and all $i\in\mathcal S$;
  \item the following inequality holds for all non-zero $u$:
    \begin{align}
      \sum_{i\in\mathcal I}C_i(u)<\sum_{i\in\mathcal I^c}C_i(u),~\forall \mathcal I \in\mathbb C,\label{eq:sufficiency}
    \end{align}
  \end{enumerate}
  then the estimator $g$ is robust.
  \label{theorem:sufficient}
\end{theorem}

\begin{proof}
  Our goal is to prove that there exists a $\beta(z)$, such that for any $t\geq \beta(z)$, $\|u\| = 1$, $a\in\mathcal A$, the following inequality holds:
  \begin{align}
    \sum_{i\in\mathcal S}f_i(y_i - H_i\times tu) < \sum_{i\in\mathcal S} f_i(y_i-H_i\times (t+1)u).
    \label{eq:sufficientdiff}
  \end{align}
  As a result, any point $\|\hat x\| \geq \beta(z)+1$ cannot be the solution of the optimization problem since there exists a better point $(\|\hat x\|-1)\hat x/\|\hat x\|$. Therefore, we must have $\|g(y)\|\leq \beta(z) + 1$ and hence the estimator is robust.

  Suppose the set of malicious sensors is $\mathcal I$, to prove \eqref{eq:sufficientdiff}, we will first look at benign sensors. Due to the uniform convergence of $h_i(u,v,t)$ to $C_i(u)$ on $\mathcal U\times\{-z_i\}$ shown in Lemma \ref{lemma:1}, given any $\delta>0$ we can always find a finite constant $N_i$ depending on $\delta$ and $z_i$ such that for all $t \geq N_i(\delta,z_i)$, the following inequality holds:
  \begin{align}
    h_i(-z_i,u,t) = \frac{1}{t}\left[f_i(tH_iu-z_i)-f_i(-z_i)\right]\geq C_i(u) - \delta,
    \label{eq:deltaapprox2}
  \end{align}
  for any $\|u\|=1$. By \eqref{eq:marginalincrease}, we can derive that
 \begin{align}
    f_i((t+1)H_iu-z_i)-f_i(tH_iu-z_i)\geq C_i(u) - \delta.
    \label{eq:deltaapprox3}
  \end{align}
  We define $\beta(z) \triangleq \max_{1\leq i\leq m} N_i(\delta,z_i)$ and fix $\delta$ to be
  \begin{align}
    \delta = \frac{1}{m}\min_{\|u\|=1}\min_{\mathcal I\in\mathbb C}\left(\sum_{i\in\mathcal I^c}C_i(u) - \sum_{i\in\mathcal I}C_i(u)\right).\label{eq:defdelta}
  \end{align}
  Notice that we write $\min_{\|u\|=1}$ instead of $\inf_{\|u\|=1}$ since $C_i(u)$ is continuous and the set $\{u:\|u\|=1\}$ is compact. Hence, the infimum is achievable, which further proves that $\delta > 0$ is strictly positive. Hence, for $i = 1,\dots,m$, if $t > \beta_{\delta}(z)$ we have
  \begin{multline}
    f_i((t+1)H_iu-z_i)-f_i(tH_iu-z_i)\\
    \geq C_i(u) - \delta,\,\forall \|u\|=1.
  \end{multline}
  Since for good sensors, $z_i = y_i$, we know that
  \begin{align}
    \sum_{i\in \mathcal I^c} & \left[f_i((t+1)H_iu-z_i)-f_i(tH_iu-z_i)\right]\nonumber\\
                             &\geq  \sum_{i\in \mathcal I^c} C_i(u) -(m-p) \delta,\,\forall \|u\|=1.
                               \label{eq:suff2}
  \end{align}
  We now consider malicious sensors. By Lemma \ref{lemma:1} (iii), we know that for $i\in\mathcal I$, and any $u$
  \begin{align}
    \sum_{i\in\mathcal I} f_i(y_i - t H_i u)  - \sum_{i\in\mathcal I} f_i(y_i- (t +1)H_i u)\leq \sum_{i\in\mathcal I}C_i(-u).\label{eq:suff1}
  \end{align}
  Hence from \eqref{eq:defdelta},~\eqref{eq:suff1} and~\eqref{eq:suff2}, we know that
  \begin{align*}
    \sum_{i\in\mathcal S} f_i(y_i &- (t+1) H_i u) - \sum_{i\in\mathcal S} f_i(y_i - t H_i u) \\
                                  &\geq \sum_{i\in\mathcal I^c}C_i(u) - \sum_{i\in\mathcal I}C_i(u) - (m-p)\delta > 0,
  \end{align*}
  which proves \eqref{eq:sufficientdiff}.
\end{proof}
\begin{remark}
  Assuming that $y_i$ is a scalar and $w=0$, Fawzi et al.~\cite{Fawzi2012} prove that the state can be exactly recovered under the integrity attack if and only if for all $u\neq 0$, there are at least $2p+1$ non-zero $H_iu$. Notice that if for some $u\neq 0$, there are less than $2p+1$ non-zero $H_iu$, then we can choose $\mathcal I$ to contain the largest $p$ $H_iu$ and thus violate \eqref{eq:sufficiency}. As a result, our sufficient condition is stronger than the ones proposed in \cite{Fawzi2012}. The main reason is that we seek to use convex optimization to solve the state estimation problem, while in \cite{Fawzi2012}, a combinatorial optimization problem is needed to recover the state.
\end{remark}
We next give necessary conditions for the robustness of the estimator.
\begin{theorem}[Necessary Condition I]
  If $ C_i(u)$ is well defined for all $u\in \mathbb R^{n}$ and all $i\in\mathcal S$ but there exist some $\|u_0\|=1,~\mathcal I_0\in\mathbb C$ such that
  \begin{align}
    \sum_{i\in\mathcal I_0}C_i(u_0)>\sum_{i\in\mathcal I_0^c}C_i(u_0),\label{eq:necessity}
  \end{align}
  then the estimator is not robust to the attack.
  \label{theorem:necessity1}
\end{theorem}
\begin{proof}
  The robustness of the estimator is equivalent to that the optimal estimate $\hat x$ satisfies $\norm{\hat x}\leq \mu(z)$ for all $a\in \mathcal A$, where $\mu$ is a real-valued function. To this end, we will prove that for any $r > 0$, there exists a $y$ such that all $\hat x$ that satisfies $\norm{\hat x}\leq r$ cannot be the optimal solution of \eqref{eq:general}.

  We will first look at the compromised sensors. For every $\delta>0$ we can always find a finite constant $N_i(\delta)$ such that for any $\hat x\in\{\hat x:\norm{\hat x}\leq r\}$ and for all $t > N_i$, the following inequality holds:
  \begin{align}
    &f_i(t H_iu_0-H_i\hat x)-f_i(tH_iu_0 - H_i(\hat x+u_0)) \nonumber\\
    \geq &f_i((t+1) H_iu_0-H_i(\hat x+u_0))-f_i(tH_iu_0 - H_i(\hat x+u_0)) \nonumber\\
    \geq & h_i(u_0,-H_i(\hat x+u_0),t)\geq C_i(u_0)-\delta,~\forall i\in \mathcal I_{0}. \label{eq:nece1}
  \end{align}
  The first inequality is derived from \eqref{eq:marginalincrease}. The second inequality is due to the uniform convergence of $h_i(u,v,t)$ to $C_i(u)$ on $\{u_0\}\times \{v:v=-H_ix+u_0,\,\|x\|\leq  r\}$.

  Let us choose
  \begin{align*}
    \delta = \frac{1}{m}\left(\sum_{i\in\mathcal I_{0}}C_i(u_0) - \sum_{i\in \mathcal I_{0}^c}C_i(u_0)\right),
  \end{align*}
  and $t = \max_{i\in\mathcal I_0} N_i(\delta)$ and $y_i = tH_iu_0$ for all $i\in \mathcal I_0$, then we know for any $\|\hat x\|\leq  r$,
  \begin{multline*}
    \sum_{i\in\mathcal I_0}\left[f_i(y_i-H_i\hat x)-f_i(y_i - H_i(\hat x+u_0))\right] \\
    \geq \sum_{i\in \mathcal I_0} C_i(u_0)-p\delta.
  \end{multline*}
  Now let us look at the benign sensors. By Lemma \ref{lemma:1} (iii) we have
  \begin{multline}
    f_i(z_i-H_i(\hat x+u_0))-f_i(z_i-H_i\hat x) \\
    \leq  C_i(u_0),~\forall i\in \mathcal I_0^c.\label{eq:nece2}
  \end{multline}
  From \eqref{eq:nece1} and \eqref{eq:nece2},
  \begin{multline*}
    \sum_{i\in\mathcal S} f_i(y_i-H_i(\hat x+u_0))  - \sum_{i\in\mathcal S} f_i(y_i-H_i\hat x)\\
    \leq \sum_{i\in\mathcal I_0^c}C_i(u_0) - \sum_{i\in\mathcal I_{0}}C_i(u_0) + p\delta < 0.
  \end{multline*}
  Thus for such a $y$ satisfying
  \[ y_i=
    \left\{
      \begin{array}{ll}
        z_i, & \hbox{if } i\in\mathcal I_0^c\\
        tH_iu_0, & \hbox{if } i\in\mathcal I_{0},
      \end{array}
    \right.
  \]
  $\hat x+u_0$ is a better estimate than all $\hat x$ satisfying $\norm{\hat x}\leq  r$. Since $r$ is an arbitrary positive real number, we can conclude that the estimator is not robust.
\end{proof}

Before continuing on, we would like to provide some remarks on the main result. First, it is worth noticing that the existence of a well defined limit of $f_i(tH_iu)/t$ is crucial for the robustness of $g$. For example, the least squares estimator cannot be robust since $f_i$ is in quadratic form. Using the potential field and force analogies in Remark~\ref{remark:force}, one can interpret the results presented in this section as: the estimator $g$ is robust if the force generated by any sensor is bounded and if the combined force of any collection of $p$ sensors is no greater than the combined force of the remaining $m-p$ sensors.

Secondly, one can see that the conditions proved in Theorem~\ref{theorem:sufficient} and \ref{theorem:necessity1} are very tight, with only a trivial gap where the LHS of \eqref{eq:necessity} equals the RHS.

\section{Concluding Remarks}\label{section:conlusion}
We have studied the robust estimation problem where $p$ out of $m$ sensors are under attack. The malicious measurements can be arbitrarily manipulated and thus a robust estimator which can give a reliable estimate is needed. Our interest is not to study any concrete estimator in presence of attacks. Instead, we have considered a general class of estimators which integrate a large number of important estimators as special cases and given sufficient and necessary conditions for the robustness of the estimator. Future works include the robustness analysis for the dynamical state estimation problem.
\section{Appendix}
\subsection*{Proof of Lemma \ref{lemma:1}:}

  \begin{enumerate}[(i)]
  \item If $\alpha = 0$, then clearly $C_i(0)= 0$. On the other hand, if $\alpha\neq 0$, from the definition in \eqref{eq:sublinear}, we have
    \begin{align*}
      C_i(\alpha u) &= \lim_{t\rightarrow\infty} \frac{1}{t}f_i(|\alpha|t H_i u)\\
                    &= \card{\alpha} \lim_{t\rightarrow\infty} \frac{1}{\card{\alpha} t}f_i(|\alpha| t H_i u)= \card{\alpha} C_i(u).
    \end{align*}
    Due to the scaling property of $C_i(u)$ and the convexity of $f_i$, we have
    \begin{align*}
      C_i(u_1+u_2)= 2C_i\left(\frac{u_1+u_2}{2}\right)\leq C_i(u_1) + C_i(u_2).
    \end{align*}
    Therefore, we know that $C_i$ is actually a semi-norm on $\mathbb R^n$
  \item Based on the convexity of $f_i$, we obtain
    \begin{align}
      2f_i(\frac{tH_iu}{2})&\leq f_i(v+tH_iu) + f_i(-v),\label{eq:temp1}\\
      f_i(tH_iu)&\geq 2f_i(\frac{2v+tH_iu}{2}) - f(2v).\label{eq:temp2}
    \end{align}
    Dividing both sides of \eqref{eq:temp1} and \eqref{eq:temp2} by $t$ and taking limit over $t$, we have
    \begin{align}
      C_i(u)&\leq \liminf_{t\rightarrow\infty}\frac{1}{t}f_i(v+tH_iu) +   \lim_{t\rightarrow\infty}\frac{1}{t}f_i(-v),\label{eq:liminf}\\
      C_i(u)&\geq \limsup_{t\rightarrow\infty}\frac{2}{t}f_i(v+\frac{t}{2}H_iu) -   \lim_{t\rightarrow\infty}\frac{1}{t}f_i(2v).\label{eq:limsup}
    \end{align}
    Since $\lim_{t\rightarrow\infty}f_i(-v)/t=\lim_{t\rightarrow\infty}f_i(2v)/t=0$, from \eqref{eq:limsup} and \eqref{eq:liminf} we have the following pointwise limit
    \begin{align*}
      \lim_{t\rightarrow\infty}h_i(u,v,t) = C_i(u).
    \end{align*}
    Notice that for a fixed $(u,v)$, by Lemma~\ref{lemma:convex}, $h(u,v,t)$ is monotonically non-decreasing with respect to $t$. Furthermore, $C_i(u)$ is continuous since it is a semi-norm. Therefore, by Dini's theorem~\cite{rudin1964principles}, $h(u,v,t)$ converges uniformly to $C_i(u)$ on a compact set of $(u,v)$.
  \item By the convexity of $f_i$, we have
    \begin{multline*}
        f_i(v + t H_iu) - f_i(v + (t-1) H_iu)\\
        \leq f_i(v + (t+1)H_iu) - f_i(v + t H_iu),
    \end{multline*}
and
    \begin{multline*}
    f_i(v + (t+1)H_iu) - f_i(v + t H_iu) \\
\leq \frac{1}{t}(f_i(v+tH_iu) - f_i(v)).
\end{multline*}
      Then we can conclude that
    \begin{multline*}
f_i(v + H_iu) - f_i(v)\\
\leq \lim_{t\rightarrow\infty} \frac{1}{t}(f_i(v+tH_iu) - f_i(v)) =C_i(u).
    \end{multline*}
  \end{enumerate}
  \hspace*{\fill}~\QED\par\endtrivlist\unskip

\bibliographystyle{IEEEtran}
\bibliography{reference}
\end{document}

%% file: intro.tex
\section{Introduction}

The concept of networks has been increasingly prevailing for decades, e.g., computer networks, sensor networks or social networks. Regardless of numerous benefits introduced by bridging machines or humans through networks, the interconnect and distributed nature renders networks vulnerable to various kinds of attacks, ranging from physical attacks to internet viruses to groundless rumors through online social networks. This article is concerned with the integrity attacks in sensor networks which are widely embedded in various industrial systems such as smart grid~\cite{MassoudAmin2005} or Supervisory Control And Data Acquisition (SCADA) systems~\cite{boyer2002scada}. During the integrity attack, the adversary can take full control of a subset of sensors and arbitrarily manipulate their measurements. The motivations for launching such an attack in industrial systems may include creating arbitrage opportunities in electricity market, stealing gas or oil without being noticed, posing potential threat to national defense, etc. Since the first SCADA system malware (called Stuxnet) was discovered and extensively investigated~\cite{Chen2010,Fidler2011}, increasing research attention has been paid to resolve the security issues in estimation and control systems~\cite{challengessecurity}.

In this article, we focus on the problem of robust estimation against compromised sensory data in order to mitigate the damage caused by the integrity attack. Robustness for an estimator is urgently needed since quite a number of the commonly used estimators under attack fail to give a reliable estimate and thus lead to poor system performance. For instance, a linear estimator is not robust since one bad measurement is enough to ruin the final estimate. A better estimator may be the geometric median of all measurements~\cite{Lopuhaa1991}. To be concrete, we consider the problem of estimating a vector state $x\in\mathbb R^n$ from measurements collected by $m$ sensors, where the measurements are subject to any random noise. For practical reasons, the spatially distributed sensors cannot be fully guaranteed to be secure. Some of them may be controlled by the attacker and due to the resource limitation the attacker can only attack up to $p<m$ sensors. Without posing any restrictions on the attacker, we assume that the compromised sensory data can be arbitrarily changed.

\emph{Related Work}: A quite similar problem in the context of power systems is bad data detection, which has been studied over the past decades~\cite{Handschin1975,Mili1985}. The method of checking the magnitude of residue is useful for identifying random bad data or outliers but may not work for intentional integrity attacks \cite{henrik2010,Xie2011}. For example, Liu et al.~\cite{liu2009} successfully showed that a stealthy attack changing the state while not being detected is possible. Kim et al.~\cite{Kim2014} studied a so-called framing attack. Under such a attack, the bad data detector is misled to delete those critical measurements, without which the network is unobservable and a convert attack may be launched.

For dynamical systems, detecting malicious components via fault detection and isolation based methods has also been extensively studied,  \cite{fp-ab-fb:09b,Pasqualetti2011,wirelesscontrol,Fawzi2012,chong2015observability}. However, in most of these works, the system is assumed to be noiseless, which greatly favors the failure detector. Pajic et al.~\cite{Pajic2014} improved the work by considering the systems with bounded noise. On the top of sufficient conditions for exact recovery in noiseless case, they showed that the worst error is still bounded even under attack. However, their estimator is based on a combinatorial optimization problem, which in general is computational hard to solve and may not be applicable for large scale systems. In \cite{Mo2010,moscs10security}, the authors use reachability analysis and ellipsoid approximation to characterize all possible biases the adversary can inject to the system.

In the area of statistics, the concept of robust estimators is not new~\cite{Kassam1985,robust2006,robust2009}. The robustness is often measured by breakdown points~\cite{Hampel1971,donoho1983notion} or influence functions~\cite{Hampel1974}. Many existing works studied one or several estimators and discussed the breakdown point properties \cite{Yohai2012,Rousseeuw1992,Hossjer1994,Rousseeuw2012}. However, a unified analysis for most useful estimators is still absent.

Motivated by different behaviors of various estimators under the integrity attacks, we manage to provide a unified robustness analysis framework integrating most commonly used estimators. To reach this goal, we first give a formal definition on the robustness of an estimator. To achieve greater generality, a general convex optimization based estimator is proposed and necessary and sufficient conditions on the robustness of such an estimator is proved. The significance of this work is that the analytical results presented in this manuscript can be used for characterizing and designing a robust estimator in the presence of compromised sensory data.

The rest of the paper is organized as follows. In Section \ref{section:problemsetup} we formulate the robust estimation problem. Our main results on the robustness of a general convex optimization based estimator is presented in Section \ref{section:main}. The concluding remarks are given in Section \ref{section:conlusion}.
